\documentclass[11pt,article]{IEEEtran}
\usepackage{amsmath,amssymb}
\usepackage{amsthm}
\usepackage{mathtools}
\usepackage{color}
\usepackage{algorithm}
\usepackage{algpseudocode}
\usepackage{float}
\newtheorem{theorem}{Theorem}
\newtheorem{lemma}{Lemma}

\newtheorem{assumption}{Assumption}

\newtheorem{example}{Example}
\newtheorem{proposition}{Proposition}

\theoremstyle{remark}
\newtheorem{remark}{Remark}

\newcommand{\sstate}{\mathbf{s}}
\newcommand{\shat}{\hat{\mathbf{s}}}
\newcommand{\shathat}{\hat{\hat{\mathbf{s}}}}
\newcommand{\Sigmahat}{\hat{\Sigma}}

\newcommand{\rvx}{\mathbf{x}}
\newcommand{\rvy}{\mathbf{y}}
\newcommand{\rvz}{\mathbf{z}}
\newcommand{\rvv}{\mathbf{v}}
\newcommand{\rvr}{\mathbf{r}}
\newcommand{\rvu}{\mathbf{u}}
\newcommand{\rve}{\mathbf{e}}
\newcommand{\rvm}{\mathbf{m}}
\newcommand{\rvmsg}{\mathbf{\theta}}

\newcommand{\rvxn}{\rvx^n}
\newcommand{\rvyn}{\rvy^n}
\newcommand{\rvzn}{\rvz^n}
\newcommand{\rvvn}{\rvv^n}

\DeclareMathOperator{\Range}{Range}
\DeclareMathOperator{\Schur}{Schur}

\onecolumn

\title{Feedback Capacity of Stationary Gaussian Channels: 
An Optimal Schalkwijk-Kailath Scheme}

\author{
David Fay and Oron Sabag%
\thanks{D. Fay and O. Sabag are with the Rachel and Selim Benin School of
Computer Science and Engineering, The Hebrew University of Jerusalem,
Jerusalem, Israel.
Emails: \texttt{david.fay@mail.huji.ac.il} and
\texttt{oron.sabag@mail.huji.ac.il}.}%
\thanks{This research was supported by the Israel Science Foundation
(grant No.~1096/23).}
}

\begin{document}
\maketitle

\begin{abstract}
We consider channels with additive colored Gaussian noise and noiseless feedback. Kim’s seminal work derived a stationary variational characterization of feedback capacity and further asserted that the capacity-achieving stationary input need not contain a feedback-independent Gaussian component. These results led to the construction of a simple coding scheme, based on the  Schalkwijk--Kailath (SK) refinement principle, which was shown to be capacity-achieving. A recent note identified a gap in the proof of the feedback-independent component-removal assertion, thereby leaving the optimality of the SK scheme and subsequent results that rely on it incomplete. In this paper, we prove the component-removal assertion for channels with stationary Gaussian noise that has a rational power spectral density. Our proof uses a perturbation analysis of a convex optimization formulation of feedback capacity and, indeed, shows that \emph{every} optimizer assigns zero power to the feedback-independent component. Using this stronger property, we construct from any optimizer an explicit SK coding scheme that achieves every rate below feedback capacity with doubly-exponentially decaying maximal error probability.

\end{abstract}

\section{Introduction}
\label{sec:intro}
We consider additive Gaussian channels
\begin{equation}
    \rvy_i=\rvx_i + \rvz_i,\qquad i=1,2,\ldots,
    \label{eq:channel}
\end{equation}
where \(\rvx_i \in \mathbb R\) is the channel input at time $i$, \(\rvy_i\) is the
channel output at time $i$, and \(\{\rvz_i\}\) is a Gaussian noise process. The encoder has noiseless output feedback so that the input $\rvx_i=f_i(\rvmsg,\rvy^{i-1})$ is a function of past channel outputs \(\rvy^{i-1}:=(\rvy_1,\ldots,\rvy_{i-1})\) and a message \(\rvmsg \sim U[1:2^{nR}]\), where $R$ is the rate and $n$ is the blocklength. The encoder is subject to the average-power constraint $ \frac1n \sum_{i=1}^n \mathbb E[\rvx_i^2]\le P$, and the decoder output is a message estimate $\hat \rvmsg=g_n(\rvyn)$. The definitions of average probability of error, achievable rates and feedback capacity are standard, e.g.,~\cite[Section I]{Kim2010}.

In this paper, we focus on additive colored Gaussian noise (ACGN) channels where the noise process $\{\rvz_i\}$ is not white. In this case, feedback capacity can exceed the no-feedback capacity \cite{Butman1969}, while for the additive white Gaussian noise (AWGN) channel, feedback does not increase capacity \cite{Shannon1956,1054671}. 

\subsection{Feedback capacity of the ACGN channel}
The feedback capacity of the ACGN channel has been studied in numerous papers, e.g. \cite{Kim2010,SabagISIT,SabagKostinaHassibi2023,CoverPombra1989,Gattami2019, Butman1969,Butman1976,Tiernan1974,Ebert1970,Shahar-Doron2004,LiElia,LiuHan2019,Elia2004,Ozarow1990,Yang2007,Kim2006,RawatElia2021,Ihara2012}. Cover and Pombra studied the ACGN with an arbitrary Gaussian noise process, whose covariance matrix is $K_{\rvz^{n}}\succ0$, and derived its feedback capacity \cite{CoverPombra1989}.
In particular, they derived the $n$-block feedback capacity optimization 
\begin{align}\label{eq:CP}
    C_n(P) :=
    \max_{\substack{K_{\mathbf{v}^{n}}\succeq0,\;B_n\ \mathrm{strictly\ lower\ triangular}\\
    \operatorname{Tr}
    (K_{\mathbf{v}^{n}}+B_nK_{\rvz^{n}}B_n^\top)\le nP}}
    \frac{1}{2n}
    \log
    \frac{
    \det\!\left(K_{\mathbf{v}^{n}}+(I+B_n)K_{\rvz^{n}}(I+B_n)^\top\right)}
    {\det K_{\rvz^{n}}}.
\end{align} 
For stationary noise processes, the limit exists and equals the feedback capacity, 
\begin{equation}
    C_{\mathrm{FB}}(P)=\lim_{n\to\infty} C_n(P).
    \label{eq:CP_limit}
\end{equation}

The optimization \eqref{eq:CP} corresponds to linear channel inputs
\begin{align}\label{eq:CP_input}
    \rvxn &= B_n\rvzn + \rvvn,
\end{align}
where $\rvvn\sim\mathcal N(0,K_{\mathbf v^n})$ is a zero-mean Gaussian vector independent of $\rvzn$, and $B_n\in\mathbb R^{n\times n}$ is strictly lower-triangular. Thus, $B_n\rvzn$ is a strictly causal linear function of the past channel noise. Past noise instances are available to the encoder due to the feedback since $\rvz^{i-1}=\rvy^{i-1}-\rvx^{i-1}$. Thus, $B_n\rvzn$ is regarded as the \emph{feedback-dependent} component, while $\rvvn$ is the \emph{feedback-independent} component.

For stationary Gaussian noise with power spectral density (PSD)
$S_Z(e^{j\omega})$, Kim \cite{Kim2010} proved that stationary Gaussian input laws suffice to characterize the Cover--Pombra limit
\eqref{eq:CP_limit}. This yields the frequency-domain variational characterization
\begin{align}
    C_{\mathrm{FB}}(P)
    =
    \sup_{\substack{S_V\ge0,\;B\ \mathrm{strictly\ causal}\\
    \frac1{2\pi}\int_{-\pi}^{\pi}
    \left(
    S_V(e^{j\omega})
    +
    |B(e^{j\omega})|^2S_Z(e^{j\omega})
    \right)d\omega\le P}}
    \frac12
    \int_{-\pi}^{\pi}
    \log
    \left(
    \frac{
    S_V(e^{j\omega})
    +
    |1+B(e^{j\omega})|^2S_Z(e^{j\omega})
    }
    {S_Z(e^{j\omega})}
    \right)
    \frac{d\omega}{2\pi}.
    \label{eq:Kim}
\end{align}
The same work further asserted that the optimization
\eqref{eq:Kim} can be restricted to $S_V\equiv0$ without loss of optimality, eliminating the feedback-independent Gaussian component and yielding
\begin{align}
    C_{\mathrm{FB}}(P)
    =
    \sup_{\substack{B\ \mathrm{strictly\ causal}\\
    \frac1{2\pi}\int_{-\pi}^{\pi}
    |B(e^{j\omega})|^2S_Z(e^{j\omega})\,d\omega\le P}}
    \frac12
    \int_{-\pi}^{\pi}
    \log |1+B(e^{j\omega})|^2
    \frac{d\omega}{2\pi}.
    \label{eq:Kim_reduced}
\end{align}

The reduction in \eqref{eq:Kim_reduced} served as the basis for several significant results. First, its specialization led to a closed-form characterization of the feedback capacity of first-order ARMA noise \cite[Thm.~5.3]{Kim2010}, resolving a long-standing problem. Second, Kim considered stationary Gaussian noise with rational PSDs through finite-dimensional state-space model (SSM) realizations, a class that includes finite-order ARMA processes. For such channels, Yang, Kav\v{c}i\'{c}, and Tatikonda \cite{Yang2007} characterized state-estimation-error input laws with a feedback-independent Gaussian component, leaving open whether stationarity and component removal can be imposed without loss of optimality. Using the stationary characterization and the component-removal assertion, Kim established a capacity-achieving steady-state input law 
\begin{align}\label{eq:SK_intro}
\rvx_i &= K\left( \sstate_i-\mathbb E[\sstate_i\mid\rvy^{i-1}]
\right),
\end{align}
where $\sstate_i$ is the state of the noise SSM realization and $K$ is a time-invariant row vector. The same work also described a  generalized Schalkwijk--Kailath (SK) coding scheme realization of this input law \cite[Lem.~6.1, Thm.~6.1]{Kim2010} that achieves feedback capacity. The generalized scheme essentially follows the estimation-error refinement principle of the classical SK scheme for AWGN channels \cite{SchalkwijkKailath1966,GallagerNakiboglu2010,5714269} and its extensions to ACGN by Butman and Wolfowitz \cite{Butman1969,Butman1976,Wolfowitz1975}.

The reduced formulation \eqref{eq:Kim_reduced} was also used in subsequent work. For example, Liu and Han \cite{LiuHan2019} characterized the ARMA$(k)$ feedback capacity through a system of polynomial equations, while Li and Elia \cite{LiElia} recast the reduced problem using a Youla parametrization and obtained an SK-type coding scheme. 

Derpich and {\O}stergaard \cite{DerpichOstergaard2022} identified an error in the reduction proof from \eqref{eq:Kim} to \eqref{eq:Kim_reduced} by \cite{Kim2010}. Roughly speaking, the reduction proof sought to show that for any feasible pair \((S_V,B)\), there exists a pair \((0,\widetilde B)\) whose objective value is no smaller than that of the original pair. However, the power calculation omitted a cross term that need not vanish. Consequently, the original argument does not establish the reduction \eqref{eq:Kim_reduced} or the results whose proofs depend on it, including the stated optimality of the generalized SK scheme based on
\eqref{eq:SK_intro}.

\subsection{Contribution}
Our main result is that Kim's reduction $S_V\equiv 0$ in \eqref{eq:Kim} holds for any stationary Gaussian noise with rational PSD. Rational PSDs capture ARMA noise processes (of any finite order) and can approximate arbitrarily well continuous PSDs e.g.,~\cite{Chonavel2002}. Our result thus closes the gap identified in \cite{DerpichOstergaard2022} for rational PSDs and supplies the missing component-removal premise in subsequent arguments that rely on this reduction, e.g., \cite{Kim2010,LiuHan2019,LiElia}.

Indeed, our result applies to noise processes admitting a general finite-dimensional SSM realization. Our analysis relies on a convex optimization characterization of feedback capacity. A perturbation argument shows that \emph{every} optimizer assigns zero power to the feedback-independent Gaussian component. For completeness, we also present a simple SK scheme, building on the estimation-error refinement constructions discussed above \cite{Kim2010,LiElia}. We make its connection to the convex formulation explicit by deriving the code parameters from any optimizer, and prove that the resulting scheme achieves every rate below feedback capacity with
doubly-exponentially decaying maximal error probability.

\paragraph*{Notation}
Random variables, either scalar or vector-valued, are denoted by bold
lowercase letters, while deterministic matrices are denoted by uppercase letters. For a sequence of random variables, we write $\rvx^n:=(\rvx_1,\ldots,\rvx_n)$. For a matrix $A$, $A^\top$, $A^\dagger$, and $\rho(A)$ denote its transpose, Moore--Penrose inverse, and spectral radius, respectively.
The notation $A\succeq0$ means that $A$ is positive semidefinite. 

\paragraph*{Organization} Section~\ref{sec:SDP} introduces the state-space noise model, its observer form, and the convex characterization of feedback capacity. Section~\ref{sec:main} presents the main results, Section~\ref{sec:proofs} contains their proofs, and Section~\ref{sec:conclusions} concludes the paper.

\section{State-space models and feedback-capacity optimization}\label{sec:SDP}
This section develops the state-space framework used in our analysis to model the channel noise statistics. We introduce the noise model, derive its observer form, recall the associated convex characterization of feedback capacity, and finally show that stationary Gaussian noise with rational PSD admits such a realization.

\subsection{State-space noise model}
Consider the linear SSM
\begin{align}\label{eq:ss_output}
    \sstate_{i+1} &= F \sstate_i+G \rvu_i,
    \\ \nonumber
    \rvz_i &= H \sstate_i+\rvr_i,
\end{align}
where $\{(\rvu_i,\rvr_i)\}_i$ is an i.i.d. zero-mean Gaussian sequence with a covariance matrix $    \begin{pmatrix}
        U & L\\
        L^\top & R
     \end{pmatrix}$. 
The state is $\sstate_i\in\mathbb R^d$, where $d<\infty$ is arbitrary, and the disturbances are $\rvu_i\in\mathbb R^q,\rvr_i\in\mathbb R$. The output process of the SSM is $\{\rvz_i\}$ corresponding to the Gaussian noise process in \eqref{eq:channel}. The matrices \(F\in\mathbb R^{d\times d}, G\in\mathbb R^{d\times q}, H\in\mathbb R^{1\times d}, U\in\mathbb R^{q\times q},L\in\mathbb R^{q\times 1}, R\in\mathbb R\) are assumed to be known and $R>0$. The initial state $\sstate_1$ is independent of $\{(\rvu_i,\rvr_i)\}$, and is distributed according to $\sstate_1\sim \mathcal N(0,\Sigma_1)$. An initial-covariance condition on $\Sigma_1$ is presented in the following subsection. 

We focus on noise processes with a nonwhite stationary limit. The following assumption is sufficient for the asymptotic stationarity of $\{\rvz_i\}$.
\begin{assumption}
\label{ass:F_stable}
    \(F\) is Schur, i.e., $\rho(F)<1$.
\end{assumption}   

The following assumption excludes the AWGN channel.
\begin{assumption}\label{ass:non_white}
The stationary limiting output process $\{\rvz_i\}$ is nonwhite.
\end{assumption}

\subsection{Kalman filter and observer form}
In \eqref{eq:ss_output}, the state $\sstate_i$ is hidden from both the encoder and the decoder. It is convenient to replace this
realization by an SSM in an observer form whose state is the encoder's MMSE estimate $\shat_i:=\mathbb E[\sstate_i\mid\rvz^{i-1}]$ with error covariance $\Sigma_i:=\operatorname{Cov}(\sstate_i-\shat_i)$.
The Kalman filter gives
\begin{align}
    \shat_{i+1}
    &=
    F\shat_i+K_{p,i}\rve_i,
    \qquad
    \rve_i:=\rvz_i-H\shat_i,
    \label{eq:shat_recursion}
\end{align}
where
$K_{p,i}=(F\Sigma_iH^\top+GL)\Psi_i^{-1}$ and
$\Psi_i:=H\Sigma_iH^\top+R$.
The innovations sequence $\{\rve_i\}$ is Gaussian and white, with
$\rve_i\sim\mathcal N(0,\Psi_i)$. 

Let $\Sigma\succeq0$ denote the maximal positive-semidefinite solution\footnote{Assumption~\ref{ass:F_stable} implies that $(F,H)$ is detectable and therefore there exists a maximal solution. } of the discrete algebraic Riccati equation (DARE)
\begin{align}
    \Sigma
    =
    F\Sigma F^\top+G U G^\top-K_p\Psi K_p^\top,
    \label{eq:Ricatti equation}
\end{align}
where $K_p:=(F\Sigma H^\top+GL)\Psi^{-1}$ and
$\Psi:=H\Sigma H^\top+R$ \cite{Kailath2000}. We also assume $\Sigma_1\succeq \Sigma$, which implies $\Sigma_i \to \Sigma, K_{p,i}\to K_p, \Psi_i \to \Psi$. 

From this point, we use the steady-state parameters. Operationally, the encoder can transmit a sub-linear growing prefix $\rvx_i=0$ so that the encoder and the decoder (since $\rvy_i = \rvz_i$) can apply \eqref{eq:shat_recursion} to have a common estimate $\shat_1$ at the prefix's end. The purpose of transmitting this prefix is to ensure that $\Sigma_i$ is arbitrarily close to $\Sigma$ and we can work with the steady-state observer form
    \begin{align} \label{eq:obs_form}
    \shat_{i+1} &= F\shat_i +  K_{p}\rve_i \nonumber\\
    \rvz_i&= H\shat_i + \rve_i,
\end{align}
where the innovation is i.i.d. with $\rve_i\sim N(0,\Psi)$ and $K_p,\Psi$ are evaluated at the maximal solution of \eqref{eq:Ricatti equation}. Note that the encoder has causal access to the hidden state $\shat_i$, and that the initial state $\shat_1$ is known to the encoder and the decoder. 

\subsection{Capacity optimization and optimal input law}
The feedback capacity of an ACGN channel with SSM noise admits a finite-dimensional convex characterization. We follow the notation of \cite{SabagKostinaHassibi2023,SabagISIT} specialized to the scalar case; similar scalar formulations appear in \cite{Gattami2019,RawatElia2021,SabagKostinaHassibi2021,SabagISIT}.

For the channel
\eqref{eq:channel} with noise generated by \eqref{eq:ss_output}, the
feedback capacity for a power constraint \(P\) is
\begin{align}
    C_{\mathrm{FB}}(P)
    =
    \max_{\substack{
        \Gamma\in\mathbb R^{1\times d}, \ \Sigmahat\in\mathbb S^{d}
        }}
    \quad
    &\frac12
    \log\left(
        \frac{\Psi_Y}{\Psi}
    \right)
    \label{eq:SKH_scalar}\\
    \mathrm{s.t.}\quad
    &
    \begin{pmatrix}
        P & \Gamma\\
        \Gamma^\top & \Sigmahat
    \end{pmatrix}
    \succeq0,
    \nonumber\\
    &
    \Psi_Y
    =
    P+H\Sigmahat H^\top+\Gamma H^\top+H\Gamma^\top+\Psi,
    \nonumber\\
    &
    K_Y = (F\Gamma^\top + F\Sigmahat H^\top+K_p \Psi )\Psi_Y^{-1},
    \nonumber\\
    &
    \mathcal L := \begin{pmatrix}
    F \Sigmahat F^\top+K_p\Psi K_p^\top-\Sigmahat
    &
    F\Gamma^\top+F\Sigmahat H^\top+K_p\Psi
    \\
    \Gamma F^\top+H \Sigmahat F^\top+\Psi K_p^\top
    &
    \Psi_Y
    \end{pmatrix}
    \succeq0.
    \nonumber
\end{align}
This is a convex optimization, and an optimizer $(\Gamma,\Sigmahat)$ parametrizes a capacity-achieving stationary input law
\begin{align}\label{eq:SKH_policy}
    \rvx_i
    &=
    \Gamma\Sigmahat^{\dagger}\bigl(\shat_i-\shathat_i\bigr)+\rvm_i ,
\end{align}
where $\shat_i$ is the encoder's state
estimate in \eqref{eq:shat_recursion}, $\shathat_i:=\mathbb E[\shat_i\mid\rvy^{i-1}]$ is the decoder's MMSE estimate,
which obeys
\begin{align}
    \shathat_{i+1}
    =
    F\shathat_i+K_Y(\rvy_i-H\shathat_i),
    \label{eq:shathat_recursion}
\end{align}
and its corresponding steady-state error variance is $\Sigmahat:= \text{cov}(\shat_i -\shathat_i)$. The sequence $\{\rvm_i\}$ is i.i.d., independent of the noise sequence, and satisfies $\rvm_i\sim\mathcal N(0,M)$ with $M:=P- \Gamma\Sigmahat^{\dagger}\Gamma^\top$. The first term in \eqref{eq:SKH_policy} transmits a linear function of the decoder's state-estimation error and generalizes the input law in \eqref{eq:SK_intro} for SSMs whose state is not directly available to the encoder. On the other hand, $\rvm_i$ is a Gaussian component that is independent of the feedback. 

In the state-space formulation, $M=0$ implies that the stationary input \eqref{eq:SKH_policy} does not have a feedback-independent Gaussian component. In particular, when $M=0$, subtracting \eqref{eq:shathat_recursion} from
\eqref{eq:shat_recursion} gives
\begin{align}
    \shat_{i+1}-\shathat_{i+1}
    &=
    F(\shat_i-\shathat_i)+K_p\rve_i
    -K_Y(\rvx_i+\rvz_i-H\shathat_i) \nonumber\\
    &=
    \bigl[F-K_Y(H+\Gamma\Sigmahat^\dagger)\bigr]
    (\shat_i-\shathat_i)+(K_p-K_Y)\rve_i,
    \label{eq:noise_driven_difference}
\end{align}
where we used $\rvy_i=\rvx_i+\rvz_i$, and the input law \eqref{eq:SKH_policy}. For a Riccati-tight optimizer, obtainable by Lemma~\ref{lem:maximal_solution}, Lemma~\ref{lem:active_modes} shows that the dynamics in \eqref{eq:noise_driven_difference} are stable on $\Range(\Sigmahat)$. The stationary difference $\shat_i-\shathat_i$ is therefore a mean-square limit of linear combinations of past noise innovations. Since each innovation is obtained causally from the channel noise \eqref{eq:shat_recursion}, the stationary input $\rvx_i=\Gamma\Sigmahat^\dagger(\shat_i-\shathat_i)$ is a strictly
causal linear function of the channel noise past. Thus its stationary spectral
representation has $S_V\equiv0$. Therefore, proving that an optimizer of \eqref{eq:SKH_scalar} satisfies $M=0$ establishes Kim's component-removal claim for noise with an SSM realization.

\begin{remark}
When $M=0$, the stationary law \eqref{eq:SKH_policy} contains no explicit message-bearing component, which may appear to preclude communication. However, \eqref{eq:SKH_policy} specifies only the stationary phase. The message can first be embedded using a finite number of channel inputs with a vanishing contribution to the average power, and subsequent inputs refine the decoder's estimate  while following \eqref{eq:SKH_policy}. An explicit SK scheme construction is given in Section~\ref{sec:scheme}.
\end{remark}

\subsection{Rational spectra as SSM}
\label{subsec:psd_to_state}
Consider a zero-mean stationary Gaussian process with a rational non-white PSD $S_Z(e^{j\omega})$. Its normalized outer spectral factorization can be written as
\begin{align}
    S_Z(e^{j\omega})
    =
    \Psi\left|\Phi(e^{j\omega})\right|^2,
    \qquad
    \Phi(\infty)=1,
\end{align}
where $\Psi>0$. Since $\Phi$ is rational, it admits a finite-dimensional
minimal realization
\begin{equation}
    \Phi(z)
    =
    1+H(zI-F)^{-1}G,
    \label{eq:factor_realization}
\end{equation}
where $F\in\mathbb R^{d\times d}$, $G\in\mathbb R^{d\times 1}$, and
$H\in\mathbb R^{1\times d}$. Minimality implies that $(F,G)$ is controllable
and $(F,H)$ is observable. Moreover, since $\Phi$ is a stable spectral factor,
$F$ is Schur.

Consider the SSM, also studied in \cite{Kim2010},
\begin{align}
    \sstate_{i+1}&=F\sstate_i+G\rvu_i,\nonumber\\
    \rvz_i&=H\sstate_i+\rvu_i,
    \label{eq:psd_output}
\end{align}
where \(\{\rvu_i\}_{i\geq1}\) is an i.i.d.\ process with
\(\rvu_i\sim\mathcal N(0,\Psi)\). The initial state is distributed according to
$\sstate_1\sim\mathcal N(0,\Pi)$, $\sstate_1\perp\!\!\!\perp\{\rvu_i:i\geq1\}$,
where $\Pi$ is the unique solution to $\Pi=F\Pi F^\top+G\Psi G^\top$. The realization \eqref{eq:psd_output} is a special case of
\eqref{eq:ss_output}, obtained by setting
$\rvu_i=\rvr_i=\rvu_i$ and $W=L=R=\Psi$. It satisfies
Assumptions~\ref{ass:F_stable} and~\ref{ass:non_white}. 
It can be easily checked via the $z$-transform that the output process $\{\rvz_i\}_{i\ge1}$ generated by \eqref{eq:psd_output} is zero-mean, and the transfer function from
$\rvu(z)$ to $\rvz(z)$ is $\Phi(z)$. Moreover, the initial state distribution is chosen such that the process is stationary.

\begin{example}[ARMA$(1,1)$]
Consider an ARMA(1,1) noise process
\begin{align}
    \rvz_i + \beta\rvz_{i-1} = \rvu_i+\alpha\rvu_{i-1},
    \qquad
    \rvu_i\sim\mathcal N(0,1),
\end{align}
with $|\alpha|, |\beta|<1$ and $\alpha\neq\beta$. Its PSD is
\begin{align}
    S_Z(e^{j\omega})
    =
    \left|
    \frac{1+\alpha e^{-j\omega}}
         {1+\beta e^{-j\omega}}
    \right|^2,
\end{align}
and \eqref{eq:factor_realization} is obtained by choosing
\begin{align}
    F=-\beta,\qquad G=1,\qquad H=\alpha-\beta .
\end{align}
\end{example}

\section{Main result}\label{sec:main}
In this section, we show that every optimizer of the
capacity optimization assigns zero power to the feedback-independent Gaussian component. We then use this structural property to construct an explicit SK scheme.
\begin{theorem}[Optimal input law]
\label{thm:main}
Consider the ACGN channel \eqref{eq:channel} with noise generated by
\eqref{eq:ss_output} under the standing assumptions. Then every optimal solution $(\Gamma,\Sigmahat)$ to the feedback capacity optimization \eqref{eq:SKH_scalar} satisfies $$M = P - \Gamma \Sigmahat^\dagger\Gamma^\top= 0.$$ 
Consequently, every optimizer induces a capacity-achieving stationary input law of the form
\begin{align}
\label{eq:policy_M0}
    \rvx_i = \Gamma \Sigmahat^\dagger(\shat_i - \hat{\shat}_i).
\end{align}
\end{theorem}
For the rational-spectrum realization of Section~\ref{subsec:psd_to_state}, Theorem~\ref{thm:main} establishes Kim's reduction from \eqref{eq:Kim} to \eqref{eq:Kim_reduced}, thereby closing the gap discussed in Section~\ref{sec:intro} and supplying the missing premise in subsequent arguments that rely on this reduction. Theorem \ref{thm:main} also has an operational consequence: an optimizer determines the refinement gain $\Gamma\Sigmahat^\dagger$. In the next section, we use this gain to construct an explicit generalized SK scheme that achieves capacity. 
\subsection{Capacity-achieving SK scheme}\label{sec:scheme}
We provide an overview of the SK coding scheme that achieves the feedback capacity. A message is represented by a vector of one-dimensional PAM symbols. The scheme consists of two phases: message embedding and refinement. Let
\begin{align}
    \Delta_i:=\shat_i-\shathat_i
    \label{eq:def_delta}
\end{align}
denote the mismatch between the encoder's and decoder's state estimates with $\Delta_1=0$. During the message-embedding phase, comprising the first $d$ channel uses, the channel inputs embed the PAM vector into the mismatch $\Delta_{d+1}$. Throughout this phase, the encoder's state estimate $\shat_i$ evolves according to \eqref{eq:obs_form}, while the decoder's state estimate is updated according to
\begin{align}
    \shathat_{i+1}
    =
    F\shathat_i+K_p(\rvy_i-H\shathat_i),
    \qquad i=1,\ldots,d.
    \label{eq:embedding_decoder}
\end{align}
Recall that the encoder can reproduce this recursion from the feedback. The purpose of using the same gain $K_p$ at the embedding phase is to make $\Delta_{d+1}$ a deterministic linear function of the embedded PAM vector. During the refinement phase, $i=d+1,\ldots,n$, the encoder transmits inputs using the stationary law \eqref{eq:policy_M0}, while $\shat_i$ and
$\shathat_i$ evolve according to \eqref{eq:obs_form} and
\eqref{eq:shathat_recursion}, respectively. After $n$ channel uses, the decoder uses the channel outputs to estimate the embedded PAM symbols and recovers the message by nearest-neighbor decoding.

\underline{Scheme parameters:}
Let $(\Gamma,\hat\Sigma)$ be any optimizer of \eqref{eq:SKH_scalar}.
If $\hat\Sigma$ does not satisfy the Riccati equation 
\begin{align}
    \hat\Sigma &=
    F\hat\Sigma F^\top+K_p\Psi K_p^\top
    -K_Y\Psi_YK_Y^\top,
    \label{eq:refinement_Riccati}
\end{align}
then, keeping $\Gamma$ fixed, compute the Riccati solution specified in Lemma~\ref{lem:maximal_solution} and replace $\hat\Sigma$ by this solution. By Lemma~\ref{lem:maximal_solution}, the resulting pair remains optimal; we continue to denote it by $(\Gamma,\hat\Sigma)$.

Let $\hat\Sigma=U\Lambda U^\top$ be the reduced SVD of $\hat\Sigma$, where $U\in\mathbb R^{d\times r}$ with $U^\top U = I_r$ and $\Lambda\in\mathbb R^{r\times r}\succ0$. Define
$F_p:=F-K_p(H+\Gamma\hat\Sigma^\dagger)$ and $A:=U^\top F_pU$.  Choose a real Jordan decomposition $A=TJT^{-1}$ and define $W:=UT$.




By Lemma~\ref{lem:active_modes}, $$C_{\mathrm{FB}}(P)=\sum_{j=1}^r\log\rho_j,$$ where $\rho_j>1$ denotes the eigenvalue modulus associated with the $j$th real modal coordinate of $J$. Hence, for any $R<C_{\mathrm{FB}}(P)$, choose rates $R_1,\ldots,R_r$ such that $0\le R_j<\log\rho_j$ and $\sum_{j=1}^r R_j=R$. Let $\mathcal A_{j,n}$ consist of $2^{nR_j}$ equally spaced PAM points in $[-1/2,1/2]$, including both endpoints. A message $\rvmsg=(\theta_1,\ldots,\theta_r)$ consists of $r$ sub-messages,
where $\theta_j$ is mapped to a PAM point
$a_j(\theta_j)\in\mathcal A_{j,n}$. We denote the
resulting PAM vector by $\mathbf a(\rvmsg):=(a_1(\theta_1),\ldots,a_r(\theta_r))^\top$.

The scheme is described in the following algorithm.
\begin{algorithm}[H]
\caption{Capacity-achieving SK scheme}
\label{alg:scheme}
\begin{algorithmic}[1]
\Require Blocklength $n$, message $\rvmsg$, and parameters
$(\Gamma,\Sigmahat,W,J,K_p,K_Y)$

\State \textbf{Message embedding:}
\State Compute 
\begin{align}
    \rvx^d\gets-C^\dagger W \mathbf{a}(\rvmsg), \quad C:=
    \left[(F-K_pH)^{d-1}K_p\ \cdots\ K_p\right].
    \label{eq:def_C}
\end{align}
\For{$i=1,\ldots,d$}
    \State Encoder transmits $\rvx_i$ and updates $\shat_{i+1}$ using
    \eqref{eq:obs_form}.
    \State Both terminals compute $\shathat_{i+1}$ using
    \eqref{eq:embedding_decoder}.
\EndFor

\State \textbf{Refinement:} 
\State $\mathbf w_{d+1}\gets0$.
\For{$i=d+1,\ldots,n$}
    \State Encoder transmits
    $
    \rvx_i=\Gamma\Sigmahat^\dagger(\shat_i-\shathat_i)
    $
    and updates $\shat_{i+1}$ according to \eqref{eq:obs_form}.
    \State Decoder updates   \begin{align*}
        \mathbf w_{i+1}
        &=
        \mathbf w_i
        -J^{-(i-d)}W^\dagger(K_p-K_Y)
        (\rvy_i-H\shathat_i).
        \end{align*}
    \State Both terminals update $\shathat_{i+1}$ according to
    \eqref{eq:shathat_recursion}.

\EndFor

\State \textbf{Decision:}
$\displaystyle
\hat\theta_j
=
\arg\min_{\theta_j}
|[\mathbf w_{n+1}]_j-a_j(\theta_j)|,\quad j=1,\ldots,r$.
\State \Return $\hat\rvmsg=(\hat\theta_1,\ldots,\hat\theta_r)$.
\end{algorithmic}
\end{algorithm}

\begin{proposition}[Capacity-achieving coding scheme]\label{prop:scheme}
For any $R<C_{\mathrm{FB}}(P)$, Algorithm~1 achieves rate $R$ for all
sufficiently large $n$. Moreover, there exist constants $a,b,c>0$,
independent of $n$, such that
\begin{align}
    P_{e,\max}^{(n)}
    \le
    a\exp\{-b\,2^{cn}\}.
\end{align}
\end{proposition}

\section{Proofs}\label{sec:proofs}
In this section, we prove Theorem \ref{thm:main} and Proposition \ref{prop:scheme}. 

\subsection{Notation}
Throughout the proof of Theorem~\ref{thm:main}, we express
the optimization \eqref{eq:SKH_scalar} in terms of the single decision variable
\begin{align}
    Q:=
    \begin{pmatrix}
        p & \Gamma\\
        \Gamma^\top & \Sigmahat
    \end{pmatrix}
    \succeq0.
\end{align}
Since $Q$ is a decision variable, we introduce the power constraint $\operatorname{Tr}\!\left(
Q
\begin{pmatrix}
1&0\\
0&0
\end{pmatrix}
\right)
\le P$. 
The objective can be expressed as
\begin{align}
\Psi_Y(Q)
&=
\begin{pmatrix}
1&H
\end{pmatrix}
Q
\begin{pmatrix}
1&H
\end{pmatrix}^\top
+\Psi,
\end{align}
and the Riccati LMI is written as
\begin{align}
\mathcal L(Q)
=
AQA^\top-BQB^\top+C\Psi C^\top
\succeq0,
\end{align}
with
\begin{align}
A=
\begin{pmatrix}
0&F\\
1&H
\end{pmatrix},
\qquad
B=
\begin{pmatrix}
0&I\\
0&0
\end{pmatrix},
\qquad
C=
\binom{K_p}{1}.
\end{align}
The corresponding gain is
$K_Y(Q):=(F\Gamma^\top+F\Sigmahat H^\top+K_p\Psi)\Psi_Y(Q)^{-1}$.

This formulation is equivalent to \eqref{eq:SKH_scalar}.
Indeed, if $p<P$, increasing $p$ to $P$ preserves $Q\succeq0$ and
increases only the lower-right entry of $\mathcal L(Q)$, preserving its positive semidefiniteness. Since $\Psi_Y(Q)$ strictly increases,
every optimizer satisfies $p=P$.

We also use the following shorthands for Schur complements. 
\begin{align}
\Schur(\mathcal L(Q))&:=
 F \Sigmahat F^\top+K_p\Psi K_p^\top-\Sigmahat - K_Y \Psi_YK_Y^\top  \label{eq:def_Riccati_shat}\\
 \Schur(Q)&:=p -\Gamma\Sigmahat^\dagger\Gamma^\top.\label{eq:def_schur}
\end{align}

\subsection{Proof of Theorem \ref{thm:main}}
The following lemmas are the main technical steps to prove Theorem \ref{thm:main}. 
\begin{lemma}\label{lem:schur_tight}
   If an optimal solution $Q$ satisfies $\Schur(\mathcal L(Q))=0$, then $\Schur(Q)=0$.
\end{lemma}
\begin{lemma}\label{lem:maximal_solution}
Let $Q$ be feasible. Then there exists a feasible $Q^\star\succeq Q$,
with the same $p$ and $\Gamma$, such that
$\Schur(\mathcal L(Q^\star))=0$.
\end{lemma}

The proofs of the two lemmas are given below in Section \ref{subsec:lemmas}. We proceed to prove Theorem \ref{thm:main}.
\begin{proof}[Proof of Theorem \ref{thm:main}]
Let $Q$ be an optimal solution. By Lemma \ref{lem:maximal_solution}, there exists a feasible $Q^\star$ such that $Q^\star\succeq Q$ and $\Schur(\mathcal L(Q^\star))=0$. The inequality $Q^\star\succeq Q$ implies $\Psi_Y(Q^\star) \ge \Psi_Y(Q)$ and therefore $Q^\star$ is optimal as well. Consider the chain of inequalities
\begin{align}\label{eq:contradiction_th}
    0 &\le \Schur(Q) \nonumber\\
      &\stackrel{(a)}\le \Schur(Q^\star) \nonumber\\
      &\stackrel{(b)}=0,
\end{align}
where $(a)$ follows from the variational form of the generalized Schur complement $\Schur(Q) = \min_y \binom{1}{y}^\top Q \binom{1}{y}$ and the fact that $Q^\star\succeq Q$, and $(b)$ follows from Lemma \ref{lem:schur_tight} using the fact that $Q^\star$ is an optimal solution with $\Schur(\mathcal L(Q^\star))=0$. Thus, any optimal solution $Q$ satisfies $\Schur(Q)=0$.
For optimal $Q$, $\Schur(Q)=P-\Gamma\Sigmahat^{\dagger}\Gamma^\top=M$, and thus $M=0.$
\end{proof}

\begin{proof}[Proof of Lemma~\ref{lem:schur_tight}]
The proof is by contradiction. Suppose that an optimal solution $Q$ satisfies $\Schur(\mathcal L(Q))=0$ and $\Schur(Q)>0$. We construct a perturbation $Q_\varepsilon=Q+\varepsilon\Delta Q$ that remains feasible for all sufficiently small $\varepsilon>0$ and strictly increases the objective, contradicting the optimality of $Q$.

Let $\Xi$ be the unique solution of the Lyapunov equation
\begin{align}
\label{eq:eq:Xi definition}
\Xi &= F \Xi F^\top + K_pK_p^\top,
\end{align}
and for $\varepsilon>0$, define the perturbation
\begin{align}\label{eq:def_perturb}
Q_\varepsilon:= Q + \varepsilon \Delta Q,
\qquad
\Delta Q :=
\begin{pmatrix}
0 & H\Xi\\
\Xi H^\top & -\Xi
\end{pmatrix}.
\end{align}
The objective function in \eqref{eq:SKH_scalar} for the perturbed matrix is 
\begin{align}
    \Psi_Y(Q_\varepsilon) &= (1\;\;H) Q_\varepsilon (1\;\;H)^\top
 + \Psi = \Psi_Y(Q) + \varepsilon H\Xi H^\top.
\end{align}
Lemma~\ref{lem:lyapunov} shows \(H\Xi H^\top>0\). Hence
\(\Psi_Y(Q_\varepsilon)>\Psi_Y(Q)\) for all $\varepsilon>0$.

We proceed to prove that $Q_\varepsilon$ is feasible for sufficiently small $\varepsilon$. The power constraint is unchanged, since the perturbation does not affect
the top-left entry of $Q$:
\begin{align}
\operatorname{Tr}\!\left(
Q_\varepsilon
\begin{pmatrix}
1&0\\
0&0
\end{pmatrix}
\right)
=
\operatorname{Tr}\!\left(
Q
\begin{pmatrix}
1&0\\
0&0
\end{pmatrix}
\right).
\end{align}

Next, a direct computation gives
\begin{align}
\mathcal L(Q_\varepsilon)
&=
\mathcal L(Q)
+
\varepsilon
\begin{pmatrix}
\Xi-F\Xi F^\top & 0\\
0 & H\Xi H^\top
\end{pmatrix} \nonumber\\
&=
\mathcal L(Q)
+
\varepsilon
\begin{pmatrix}
K_pK_p^\top & 0\\
0 & H\Xi H^\top
\end{pmatrix},
\end{align}
where the second equality follows from \eqref{eq:eq:Xi definition}. Since
$\mathcal L(Q)\succeq0$ and $\Xi\succeq0$, it follows that
$\mathcal L(Q_\varepsilon)\succeq0$.

It remains to prove that $Q_\varepsilon\succeq0$ for sufficiently
small $\varepsilon>0$. We first show that $\ker(Q)\subseteq\ker(Q_\varepsilon)$. Let $u\in\ker(Q)$ and write $u=\binom{a}{v}$. Since $Qu=0$, we get $Q_\varepsilon u=\varepsilon\Delta Q u$. By
Lemma~\ref{lem:kernel}, $a=0$ and $\Xi v=0$. Therefore $\Delta Q u
=
\binom{H\Xi v}{a\Xi H^\top-\Xi v}
=
0.$

Now let \(\mathcal R:=\ker(Q)^\perp\). Since \(Q\succeq0\), the restriction of \(Q\) to \(\mathcal R\) is positive definite. Hence there exists \(m>0\) such that \(r^\top Qr\ge m\|r\|^2\) for all \(r\in\mathcal R\). Also, since \(\Delta Q\) is fixed, there exists \(M_\triangle<\infty\) such that \(|r^\top\Delta Qr|\le M_\triangle\|r\|^2\) for all \(r\in\mathcal R\). For any vector \(x\), decomposed as \(x = u + r\), where \(u\in\ker(Q)\) and \(r\in\mathcal R\), consider  
\begin{align}
x^\top Q_\varepsilon x
&=(u+r)^\top Q_\varepsilon(u+r) \nonumber\\
&\stackrel{(a)}=r^\top Qr+\varepsilon r^\top\Delta Qr \nonumber\\
&\ge (m-\varepsilon M_\triangle)\|r\|^2,
\end{align}
where $(a)$ follows from \(Q_\varepsilon u=0\) and
the symmetry of \(Q_\varepsilon\). We can now conclude that $Q_\varepsilon$ is feasible for $\varepsilon \le \frac{m}{M_\triangle}$, contradicting the optimality of \(Q\). Therefore
\(\Schur(Q)=0\).
\end{proof}

\begin{proof}[Proof of Lemma \ref{lem:maximal_solution}]
Write $Q=\begin{pmatrix}p&\Gamma\\\Gamma^\top&\Sigmahat\end{pmatrix}$ and, for a
symmetric matrix $Z$, let $Q_Z:=\begin{pmatrix}p&\Gamma\\\Gamma^\top&Z\end{pmatrix}$.
We define
\begin{align}\label{eq:shift}
    N(Z):=\mathcal L\bigl(Q_{\,\Sigmahat + Z}\bigr).
\end{align}
The shift has placed its constant part exactly at the feasible point:
\begin{align}\label{eq:N_zero}
    N(0)=\mathcal L(Q)\succeq0 .
\end{align}

By \cite[Appendix~E, Lemma~E.3.2]{KailathSayedHassibi2000} there exists $Z_+\succeq0$ with $\Schur(N(Z_+))=0$, moreover it is the maximal element among matrices satisfying $\Schur(N(Z))\succeq0$.

Set $Q^\star:=Q_{\Sigmahat + Z_+}$. Then $Q^\star\succeq Q\succeq0$, and the objective does not decrease. The power constraint is unchanged, because $Q^\star$ and $Q$ have the same upper-left entry. Finally
$\mathcal L(Q^\star)=N(Z_+)\succeq0$, so $Q^\star$ is feasible, and
$\Schur(\mathcal L(Q^\star))=\Schur(N(Z_+))=0$.
Moreover, $\Sigmahat + Z_+$ is the maximal solution to the equation $\Schur(\mathcal L(Q_Z))=0$.
\end{proof}

\subsection{Technical lemmas for Theorem \ref{thm:main}}\label{subsec:lemmas}

\begin{lemma}
\label{lem:lyapunov}
Let $\Xi$ be the unique solution of the Lyapunov equation
\begin{align}\label{eq:Lyapunov_lemma}
\Xi = F \Xi F^\top + K_pK_p^\top.    
\end{align}
Then $H\Xi H^\top >0$.
\end{lemma}


\begin{proof}[Proof of Lemma \ref{lem:lyapunov}]
The equation in \eqref{eq:Lyapunov_lemma} is a Lyapunov equation with a Schur
stable $F$. Thus, it admits a unique positive semidefinite solution, given by
\begin{align}\label{eq:LYA_sol}
\Xi=\sum_{k=0}^{\infty}F^k K_pK_p^\top(F^\top)^k.
\end{align}
Write
\begin{align}
    H\Xi H^\top= \sum_{k=0}^{\infty}HF^kK_pK_p^\top(F^\top)^k H^\top,
\end{align}
so it suffices to exhibit some $k\ge0$ with $HF^kK_p\neq 0$.

Consider the steady-state observer form \eqref{eq:obs_form}. Iterating the state
recursion from the initial time and substituting into the output equation gives,
for every $i\ge1$,
\begin{align}
\rvz_i = HF^{i-1}\shat_1 + \rve_i + \sum_{k=0}^{i-2} HF^kK_p\,\rve_{i-k-1}.
\end{align}
Under the assumption $HF^kK_p=0$ for all $k\ge0$, the sum vanishes and, since $F$
is Schur stable, the transient $HF^{i-1}\shat_1\to0$ as $i\to\infty$. Hence
$\rvz_i-\rve_i\to0$, so the stationary limit of $\{\rvz_i\}$ is the white process
$\{\rve_i\}$, contradicting Assumption~\ref{ass:non_white}. Therefore
$HF^kK_p\neq0$ for some $k\ge0$, and we have $H\Xi H^\top>0$.
\end{proof}

\begin{lemma}
\label{lem:kernel}
Let $Q\succeq0$ satisfy $\Schur(Q)>0$ and $\Schur (\mathcal L(Q))=0$. Let $\Xi$ be defined by \eqref{eq:eq:Xi definition}. If a scalar $a$ and a vector $v$ satisfy $Q \binom{a}{v} = 0$, then $a=0$ and $\Xi v=0$.
\end{lemma}

\begin{proof}[Proof of Lemma \ref{lem:kernel}]
We first prove that $\ker(Q)=\{(0,v):v\in\ker(\Sigmahat)\}$. The matrix $Q$ is PSD and thus for any scalar $a$ and vector $v$, we can use the completion of the square to write 
\begin{align}\label{eq:kernel}
\binom{a}{v}^{\top}Q\binom{a}{v}
&=
\Schur(Q)a^2
+
\bigl(v+a\Sigmahat^\dagger\Gamma^\top\bigr)^\top
\Sigmahat
\bigl(v+a\Sigmahat^\dagger\Gamma^\top\bigr).
\end{align}
Now let \(\binom{a}{v}\in\ker(Q)\). Since \(\Schur(Q)>0\) and
\(\Sigmahat\succeq0\), \eqref{eq:kernel} gives \(a=0\) and
\(v\in\ker(\Sigmahat)\). Conversely, if \(v\in\ker(\Sigmahat)\), then
also \(\Gamma v=0\), and hence \(Q\binom{0}{v}=0\). 




We proceed to show that for all $v\in\ker(\Sigmahat)$ we have both $K_p^\top v=0$ and $F^\top v\in\ker(\Sigmahat)$. This, together with the kernel characterization, and the expression for $\Xi$ in \eqref{eq:LYA_sol} completes the proof that $\Xi v=0$.

The fact that $\Schur(\mathcal L(Q))=0$ implies by \eqref{eq:def_schur} that the top-left block of $\mathcal L(Q)$ satisfies $F \Sigmahat F^\top+K_p\Psi K_p^\top-\Sigmahat = K_YK_Y^\top \Psi_Y$. Recall that $ K_Y = (F\Gamma^\top + F\Sigmahat H^\top+K_p \Psi )\Psi_Y^{-1}$, therefore we can write the factorization 
    
\[
\mathcal L(Q)=
\Psi_Y
\binom{K_Y}{1}\binom{K_Y}{1}^\top.
\]

Let $v\in\ker(\Sigmahat)$ and choose $\lambda = -{K_Y^\top v}$ so that $\begin{pmatrix}
    v\\\lambda
\end{pmatrix}\in \ker(\mathcal L (Q))$.  
We get
\begin{align}
\label{seperate_vanish}
0&=(v,\lambda)^\top\mathcal L(Q)(v,\lambda)
\nonumber \\ 
&=(v,\lambda)^\top(AQA^\top-BQB^\top + C\Psi C^\top)(v,\lambda)
 \nonumber \\ 
 &=(v,\lambda)^\top(AQA^\top + C\Psi C^\top)(v,\lambda),
\end{align}
where the last equality follows from $ B^\top (v,\lambda) = (0,v) $ and the fact that $(0,v)\in \ker(Q)$. Both terms in \eqref{seperate_vanish} are nonnegative, and thus 
\begin{align}
    (v,\lambda)^\top(AQA^\top)(v,\lambda) = 0
    \\
    (v,\lambda)^\top(C \Psi C^\top)(v,\lambda) =0
\end{align}
Compute $A^\top(v,\lambda)=(\lambda,F^\top v+H^\top\lambda)$, so we have $(\lambda,F^\top v+H^\top\lambda) \in \ker(Q)$. Therefore $\lambda=0$ and $F^\top v\in\ker(\Sigmahat)$. Next, by $C^\top(v,\lambda)=K_p^\top v+\lambda$ and $\Psi>0$, we obtain $K_p^\top v=0$. Since \(F^\top v\in\ker(\Sigmahat)\), induction gives \((F^\top)^k v\in\ker(\Sigmahat)\) for every \(k\ge0\). Since
\(K_p^\top\) annihilates \(\ker(\Sigmahat)\), we get
\(K_p^\top(F^\top)^k v=0\) for every \(k\ge0\). Recalling the solution of $\Xi$ in \eqref{eq:LYA_sol} yields $\Xi v=0$ as required. 
\end{proof}

\subsection{Proof of Proposition~\ref{prop:scheme}}\label{sec:proof_prop}
The proof relies on three lemmas. Lemma~\ref{lem:active_modes} identifies the active dynamics induced by an optimizer and relates their volume expansion to feedback capacity. Lemma~\ref{lemma:embedding} shows that the message can be embedded exactly into the encoder--decoder mismatch. Lemma~\ref{lem:refinement} characterizes the effective channel seen by the decoder during refinement. We then combine these results to prove Proposition~\ref{prop:scheme}.

\begin{lemma}[Active modes]\label{lem:active_modes}
Let $(\Gamma,\hat\Sigma)$ be an optimizer of \eqref{eq:SKH_scalar}
for which the Riccati LMI is tight. Then:
\begin{enumerate}
    \item[\emph{(i)}] $F_pW=WJ$,
    \item[\emph{(ii)}] $J$ is antistable,
    \item[\emph{(iii)}] $J_Y:=W^\dagger(F-K_Y\tilde H)W$ is Schur,
    \item[\emph{(iv)}] $C_{\mathrm{FB}}(P)=\log|\det J|$.
\end{enumerate}
\end{lemma}

\begin{proof}[Proof of Lemma \ref{lem:active_modes}]
Recall that Theorem~\(\ref{thm:main}\) proves $M=0$ and therefore $P = \Gamma \Sigmahat^\dagger\Gamma^\top$. 

\noindent\emph{(i)} We first show that $\mathcal X: = \Range(\Sigmahat)$ is invariant under $F_p$. The DARE \eqref{eq:refinement_Riccati} can be expressed as 
\begin{align}\label{eq:Omega_DARE}
    \hat\Sigma
    &=F_p\Omega F_p^\top, \ \ \     \Omega
    :=
    \hat\Sigma-\hat\Sigma\tilde H^\top
    \Psi_Y^{-1}\tilde H\hat\Sigma,
\end{align}
where $\tilde H:=H+\Gamma\hat\Sigma^\dagger$. By the matrix inversion lemma, $\Omega
    =
    \hat\Sigma^{1/2}
    \left(
        I+\hat\Sigma^{1/2}\tilde H^\top
        \Psi^{-1}\tilde H\hat\Sigma^{1/2}
    \right)^{-1}
    \hat\Sigma^{1/2}$. Therefore $\Range(\Omega)=\Range(\hat\Sigma)=\mathcal X$, and by \eqref{eq:Omega_DARE} $\mathcal X
    = \Range(F_p\Omega F_p^\top)
    = F_p\mathcal X$.
Thus $\mathcal X$ is invariant under $F_p$. Since $\Range(U)=\mathcal X$, it follows that
$F_pU=UU^\top F_pU=UA$. Multiplying by $T$, and using $W=UT$ and $AT=TJ$ gives $F_pW=WJ$. Moreover, $F_p\mathcal X=\mathcal X$ implies that $J$ is invertible.

\noindent\emph{(ii)}
Set $S:=T^{-1}\Lambda T^{-\top}\succ0$ and $h:=\tilde HW$ and observe that $\hat\Sigma=WSW^\top$. Substituting this in \eqref{eq:Omega_DARE} together with $F_pW = WJ$ gives
\begin{align}\label{eq:temp}
    W S W^\top &= W J\left(S-Sh^\top\Psi_Y^{-1}hS\right)J^\top W^\top.
\end{align}
Multiplying \eqref{eq:temp} on the left by $W^\dagger$ and on the right
by $(W^\dagger)^\top$, using $W^\dagger W=I_r$, gives
$S=J(S-Sh^\top\Psi_Y^{-1}hS)J^\top$. Taking inverses and multiplying by $J^\top$ and $J$ gives 
\begin{align}
    J^\top S^{-1}J
    &= \left(S-Sh^\top\Psi_Y^{-1}hS\right)^{-1} \nonumber\\
    &= S^{-1}+h^\top\Psi^{-1}h,
    \label{eq:sch_info}
\end{align}
where the last equality follows from the matrix inversion lemma and
$\Psi_Y=\Psi+hSh^\top$.

For $Jv=\lambda v$ with $v\neq0$, \eqref{eq:sch_info} gives
$(|\lambda|^2-1)v^*S^{-1}v=\Psi^{-1}|hv|^2\ge0$. Since $S\succ0$, every eigenvalue of $J$ satisfies $|\lambda|\ge1$. It remains to exclude equality.

Suppose $|\lambda|=1$ and set $x:=Wv$. Then $\tilde Hx=0$ and by $(i)$, 
$F_px=\lambda x$ and $Fx=\lambda x$.
Detectability of $(F,H)$ implies $Hx\neq0$.
Let $D:=\operatorname{Re}(xx^*)\succeq0$. Then
$\Range(D)\subseteq\mathcal X$ since $x=Wv$, $\tilde HD=0$ by $\tilde Hx =0 $, and $FDF^\top=D$.

For sufficiently small $\varepsilon>0$, set
$\hat\Sigma_\varepsilon:=\hat\Sigma-\varepsilon D\succeq0$ and
$\Gamma_\varepsilon:=\Gamma+\varepsilon HD$.
Since $\tilde HD=0$, we have
$\Gamma_\varepsilon=\Gamma\hat\Sigma^\dagger\hat\Sigma_\varepsilon$.
The first LMI therefore remains feasible, with Schur complement
\begin{align*}
    P-\Gamma_\varepsilon\hat\Sigma_\varepsilon^\dagger
    \Gamma_\varepsilon^\top
    =
    \varepsilon|Hx|^2>0.
\end{align*}
Since $FDF^\top=D$, the upper-left block of the Riccati LMI is unchanged;
the changes in its off-diagonal blocks cancel.
Its lower-right entry increases to
$\Psi_{Y,\varepsilon}=\Psi_Y+\varepsilon|Hx|^2$.
Thus the perturbed pair is feasible and strictly improves the objective,
contradicting optimality. Hence $J$ is antistable.

\noindent\emph{(iii)}
We first show that $J_YSJ^\top=S$. Consider the following chain of equalities
\begin{align}\label{lemma:scheme_technical_iii}
    J_YSJ^\top &= W^\dagger(F-K_Y\tilde H)W SJ^\top \nonumber\\
    &\stackrel{(a)}=     W^\dagger(F-K_Y\tilde H)\hat\Sigma     F_p^\top(W^\dagger)^\top\nonumber\\
    &=  W^\dagger(F_p+(K_p-K_Y)\tilde H)\hat\Sigma     F_p^\top(W^\dagger)^\top\nonumber\\
    &\stackrel{(b)}=
    W^\dagger F_p\Omega F_p^\top(W^\dagger)^\top\nonumber\\
    &= W^\dagger\hat\Sigma(W^\dagger)^\top \nonumber\\
    &= S,
\end{align}
where $(a)$ follows from $J^\top=W^\top F_p^\top(W^\dagger)^\top$ by $(i)$ and
$WSW^\top=\hat\Sigma$, and $(b)$ follows from the Kalman gains' difference 
\begin{align}
    K_p-K_Y
    &= \left(K_p\Psi_Y - F\hat\Sigma\tilde H^\top - K_p \Psi\right)
    \Psi_Y^{-1} \nonumber\\
    &= \left(K_p (\Psi_Y - \Psi) - F\hat\Sigma\tilde H^\top \right)
    \Psi_Y^{-1} \nonumber\\
    &\stackrel{(\star)}=-F_p\hat\Sigma\tilde H^\top\Psi_Y^{-1},
    \label{eq:gain_difference}
\end{align}
and $(\star)$ follows from $\Psi_Y-\Psi=\tilde H\hat\Sigma\tilde H^\top$, implied by $M=0$. 

Since $S\succ0$ and $J$ is invertible,
\eqref{lemma:scheme_technical_iii} gives
$J_Y=SJ^{-\top}S^{-1}$. Hence $J_Y$ is similar to $J^{-\top}$ and is Schur by~(ii).

\noindent \emph{(iv)}
Taking determinants in \eqref{eq:sch_info} and applying the matrix
determinant lemma gives
\begin{align*}
    |\det J|^2
    &=
    \det\!\left(I_r+Sh^\top\Psi^{-1}h\right)\\
    &=
    1+\Psi^{-1}hSh^\top
    =
    \frac{\Psi_Y}{\Psi}.
\end{align*}
Optimality in \eqref{eq:SKH_scalar} therefore gives
$C_{\mathrm{FB}}(P)=\frac12\log(\Psi_Y/\Psi)
=\log|\det J|$.
\end{proof}

\begin{lemma}[Message embedding]\label{lemma:embedding}
For every message $\theta$, the embedding inputs
$\rvx^d=-C^\dagger W\mathbf a(\theta)$ produce
$\Delta_{d+1}=W\mathbf a(\theta)$.
\end{lemma}

\begin{proof}[Proof of Lemma \ref{lemma:embedding}]
During embedding, the mismatch satisfies
$\Delta_{i+1}=F_0\Delta_i-K_p\rvx_i$, where
$F_0:=F-K_pH$ and $\Delta_1=0$. Hence
\begin{align}
    \Delta_{d+1}
    =-C\rvx^d
    =CC^\dagger W\mathbf a(\theta).
    \label{eq:implant_projection}
\end{align}
Since $CC^\dagger$ is the orthogonal projector onto $\Range(C)$,
it remains to prove $\Range(W)\subseteq\Range(C)$ so that $CC^\dagger W = W$.

Let $\mathcal R:=\Range(C)$. By Cayley--Hamilton,
$\mathcal R=\operatorname{span}\{F_0^kK_p:k\ge0\}$.
For $v\in\mathcal R^\perp$, we have
$K_p^\top(F_0^\top)^kv=0$ for every $k\ge0$.
Since $F_p=F_0-K_p\Gamma\hat\Sigma^\dagger$, for every $k\ge0$,
\begin{align*}
    F_p^\top(F_0^\top)^kv
    &=
    (F_0^\top)^{k+1}v
    -\hat\Sigma^\dagger\Gamma^\top
      K_p^\top(F_0^\top)^kv\\
    &=
    (F_0^\top)^{k+1}v.
\end{align*}
Thus induction gives $(F_p^\top)^kv=(F_0^\top)^kv$.
On the other hand, $F_pW=WJ$ implies $F_p^kW=WJ^k$, so
\begin{align*}
    W^\top(F_0^\top)^kv
    =
    W^\top(F_p^\top)^kv
    =
    (J^k)^\top W^\top v.
\end{align*}
Multiplying by $(J^{-k})^\top$ gives
\begin{align}
    W^\top v
    =
    (J^{-k})^\top W^\top(F_0^\top)^kv.
    \label{eq:implant_unreachable}
\end{align}
The solution to the DARE \eqref{eq:Ricatti equation} is maximal and therefore its closed loop matrix satisfies $\rho(F_0)\le1$ and thus 
$\|F_0^k\|$ can only grow polynomially with $k$. Since $J$ is antistable, $\|J^{-k}\|$ decays exponentially with $k$, and thus the right-hand side of
\eqref{eq:implant_unreachable} tends to zero, implying $W^\top v=0$. 
Hence $\mathcal R^\perp\subseteq\ker(W^\top)$. Therefore $CC^\dagger W=W$, and
\eqref{eq:implant_projection} proves the claim.
\end{proof}

\begin{lemma}[Effective channel]\label{lem:refinement}
Conditioned on the message $\theta$, the decoder statistic is distributed as
\begin{align}\label{eq:lemma_w_law}
    \mathbf w_{n+1}
    \sim
    \mathcal N\left(
        \mathbf{a}(\theta)-J^{-(n-d)}J_Y^{(n-d)} \mathbf{a}(\theta),\,
        J^{-(n-d)}\Xi_{n-d} J^{-(n-d)\top}
    \right),
\end{align}
where $J_Y$ is defined in Lemma \ref{lem:active_modes} and $0\preceq\Xi_{n-d}\preceq W^\dagger\hat\Sigma W^{\dagger\top}$.
\end{lemma}

\begin{proof}[Proof of Lemma \ref{lem:refinement}]
We first show that the decoder statistic satisfies
\begin{align}
    \mathbf w_i
    =
    \mathbf a(\theta)-J^{-(i-d-1)}W^\dagger\Delta_i,
    \qquad i=d+1,\ldots,n+1,
    \label{eq:effective_identity}
\end{align}
and then characterize the distribution of $W^\dagger\Delta_{i}$ conditioned on $\rvmsg=\theta$.

To prove \eqref{eq:effective_identity}, recall that during refinement,
$i=d+1,\ldots,n$, the mismatch satisfies
\begin{align}
    \Delta_{i+1}
    &=(F-K_Y\tilde H)\Delta_i+(K_p-K_Y)\rve_i \nonumber\\
    &=F_p\Delta_i+(K_p-K_Y)(\rvy_i-H\shathat_i),
    \label{eq:refinement_dynamics}
\end{align}
with $\Delta_{d+1}=W\mathbf a(\theta)$ by
Lemma~\ref{lemma:embedding}.
Since $K_p-K_Y=-F_p\hat\Sigma\tilde H^\top\Psi_Y^{-1}$ has range in
the $F_p$-invariant subspace $\mathcal X=\Range(W)$,
the mismatch remains in $\mathcal X$.
Thus $\Delta_i=WW^\dagger\Delta_i$ and, using $F_pW=WJ$,
$W^\dagger F_p\Delta_i=JW^\dagger\Delta_i$.

Multiplying the second line of \eqref{eq:refinement_dynamics} by
$J^{-(i-d)}W^\dagger$ and adding the decoder update shows that
$\mathbf w_i+J^{-(i-d-1)}W^\dagger\Delta_i$ is independent of $i$.
Its value is $\mathbf a(\theta)$, since $\mathbf w_{d+1}=0$
and $W^\dagger\Delta_{d+1}=\mathbf a(\theta)$.
This proves \eqref{eq:effective_identity}.

It remains to characterize the distribution of $W^\dagger\Delta_i$. For $i=d+1,\ldots,n$, the first line of \eqref{eq:refinement_dynamics} gives
\begin{align}
    W^\dagger\Delta_{i+1}
    =
    J_YW^\dagger\Delta_i+W^\dagger(K_p-K_Y)\rve_i.
    \label{eq:stable_refinement}
\end{align}
Since $W^\dagger\Delta_{d+1} = \mathbf a(\theta)$ and the innovations are i.i.d. Gaussian and independent of the message, \eqref{eq:stable_refinement} implies by induction that, conditioned on $\rvmsg=\theta$, $W^\dagger\Delta_i$ is Gaussian with mean
$J_Y^{i-d-1}\mathbf a(\theta)$.

Let $\Xi_{i-d-1}:=\operatorname{Cov}(W^\dagger\Delta_i\mid\rvmsg=\theta)$.
By \eqref{eq:implant_projection}, $\Xi_0=0$.
Using the independent innovations in \eqref{eq:obs_form}, the recursion
\eqref{eq:stable_refinement} gives
\begin{align}\label{eq:refinement_covariance}
    \Xi_{i-d}
    &=
    J_Y\Xi_{i-d-1}J_Y^\top
    +W^\dagger(K_p-K_Y)\Psi(K_p-K_Y)^\top(W^\dagger)^\top
    \nonumber\\
    &= J_Y\Xi_{i-d-1}J_Y^\top + W^\dagger (\hat\Sigma
-  (F-K_Y\tilde H)\hat\Sigma(F-K_Y\tilde H)^\top) (W^\dagger)^\top,
\end{align}
where the second equality follows from the DARE \eqref{eq:refinement_Riccati} that can be expressed as
\begin{align}
    \hat\Sigma
    &= (F-K_Y\tilde H)\hat\Sigma(F-K_Y\tilde H)^\top  +
    (K_p-K_Y)\Psi(K_p-K_Y)^\top.
\end{align}

Using $\hat\Sigma=WW^\dagger\hat\Sigma(W^\dagger)^\top W^\top$
and the definition of $J_Y$, \eqref{eq:refinement_covariance} gives
\begin{align*}
    W^\dagger\hat\Sigma(W^\dagger)^\top-\Xi_{i-d}
    =
    J_Y\left(
        W^\dagger\hat\Sigma(W^\dagger)^\top-\Xi_{i-d-1}
    \right)J_Y^\top.
\end{align*}
Since $\Xi_0=0$, induction yields
$\Xi_{i-d-1}\preceq W^\dagger\hat\Sigma(W^\dagger)^\top$
for $i=d+1,\ldots,n+1$. Combining the covariance with the derived conditional mean and \eqref{eq:effective_identity} at $i=n+1$ proves \eqref{eq:lemma_w_law}.
\end{proof}

\begin{proof}[Proof of Proposition~\ref{prop:scheme}]
By nearest-neighbor decoding, the sub-message $\theta_j$ is decoded incorrectly only if the estimation error has magnitude at least half the PAM spacing. Consider a sub-message with $R_j>0$, and let $E_{j,n}:= [\mathbf w_{n+1}]_j-a_j(\theta_j)$ denote the estimation error. By Lemma~\ref{lem:refinement}, $E_{j,n}$ conditioned on $\theta$, is Gaussian with mean and variance
\begin{align}
    \mu_{j,n}
    &:=-e_j^\top J^{-(n-d)}J_Y^{(n-d)}\mathbf a(\theta), \nonumber\\
    \sigma_{j,n}^2
    &:=e_j^\top J^{-(n-d)}\Xi_{n-d}(J^{-(n-d)})^\top e_j .
    \label{eq:coordinate_mean_variance}
\end{align}

Let $d_{j,n}$ denote the PAM spacing distance. We can bound the error probability for message $\theta_j$ as follows
\begin{align}
    \Pr(\hat\theta_j\neq\theta_j\mid\rvmsg=\theta)
    &\le
    \Pr\left(
        \left|E_{j,n}\right|
        \ge \frac{d_{j,n}}{2}
        \,\middle|\,\rvmsg=\theta
    \right) \nonumber\\
        &\le
    \Pr\left(
        |E_{j,n}-\mu_{j,n}|
        \ge
        \frac{d_{j,n}}2-|\mu_{j,n}|
        \,\middle|\,\rvmsg=\theta
    \right)\nonumber\\
    &\le
    2\exp\left\{
        -\frac{\left(\frac{d_{j,n}}{2}-|\mu_{j,n}|\right)^2}
        {2\sigma_{j,n}^2}
    \right\},
    \label{eq:coordinate_error}
\end{align}
where the last inequality follows from the Gaussian tail bound whenever $\frac{d_{j,n}}2-|\mu_{j,n}| >0$, and we also use the fact that for sufficiently large $n$, $\sigma_{j,n}>0$, since
$\Xi_{n-d}\to W^\dagger\hat\Sigma(W^\dagger)^\top\succ0$
and $J$ is invertible.

We next analyze the asymptotic behavior of the argument of the exponent, and then show that the condition above is satisfied for large $n$. By the Jordan form of $J$, for some constant $c>0$ independent of $n$,
\begin{align}
    \|e_j^\top J^{-(n-d)}\|
    \le
    c\,n^{r-1}\rho_j^{-(n-d)}.
\end{align}
Since $J_Y$ is Schur, $\mathbf a(\theta)$ is uniformly bounded, and
$\Xi_{n-d}\preceq W^\dagger\hat\Sigma(W^\dagger)^\top$, it follows
uniformly in $\theta$ that
\begin{align}
    |\mu_{j,n}|,\ \sigma_{j,n}
    \le
    c\,n^{r-1}\rho_j^{-(n-d)}.
    \label{eq:error_scale}
\end{align}
Moreover, since $\mathcal A_{j,n}$ contains $2^{nR_j}$ equally spaced
points in $[-1/2,1/2]$, $d_{j,n}\ge 2^{-nR_j}$. Hence, for sufficiently large $n$,
\begin{align}
    \frac{\frac{d_{j,n}}2-|\mu_{j,n}|}{\sigma_{j,n}}
    &\ge
    \frac{
        \frac12 2^{-nR_j}
        -c n^{r-1}\rho_j^{-(n-d)}
    }{
        c n^{r-1}\rho_j^{-(n-d)}
    } =
    \frac{
        2^{\,n(\log\rho_j-R_j)-d\log\rho_j}
    }{2c n^{r-1}}-1 .
    \label{eq:normalized_distance}
\end{align}
Therefore, if $R_j<\log\rho_j$, the right-hand side grows exponentially
with $n$. In particular, for any
$0<\delta_j<\log\rho_j-R_j$ and all sufficiently large $n$,
$ \frac{\frac{d_{j,n}}2-|\mu_{j,n}|}{\sigma_{j,n}}
    \ge 2^{\delta_j n}$. This also verifies $\frac{d_{j,n}}2>|\mu_{j,n}|$ for sufficiently large
$n$. Substituting into \eqref{eq:coordinate_error} gives $\Pr(\hat\theta_j\neq\theta_j\mid\rvmsg=\theta)
    \le 2\exp\left\{-\frac12\,2^{2\delta_j n}\right\}$.

Taking $\delta:=\min_j\delta_j>0$ and applying a union bound over the
$r$ coordinates gives $P_{e,\max}^{(n)} \le 2r\exp\left\{-\frac12 2^{2\delta n}\right\}$, which is doubly exponential in $n$. By Lemma~\ref{lem:active_modes}, $\sum_j\log\rho_j=C_{\mathrm{FB}}(P)$. Hence, for every $R<C_{\mathrm{FB}}(P)$, we can choose rates $R_j<\log\rho_j$ such that $\sum_jR_j=R$.

\emph{Power constraint.} The input during the refinement phase is $\rvx_{i} = \Gamma\hat\Sigma^\dagger \Delta_{i}$. By Lemma~\ref{lem:refinement}, $\Delta_{i}\in\Range(W)$, so we can write $\rvx_{i} = \Gamma\hat\Sigma^\dagger WW^\dagger\Delta_{i}$. The distribution of $W^\dagger\Delta_{i}$ in Lemma \ref{lem:refinement} implies 
\begin{align}
    \mathbb E[\rvx_{i}^2\mid\theta] &= \Gamma\hat\Sigma^\dagger W \Xi_{i-d-1} W^\top \hat\Sigma^\dagger \Gamma^\top +
    \left(\Gamma\hat\Sigma^\dagger WJ_Y^{i-d-1}\mathbf a(\theta) \right)^2\nonumber\\ 
    &\le \Gamma\hat\Sigma^\dagger W W^\dagger\hat\Sigma W^{\dagger\top} W^\top \hat\Sigma^\dagger \Gamma^\top + \left( \Gamma\hat\Sigma^\dagger WJ_Y^{i-d-1}\mathbf a(\theta)\right)^2\nonumber\\
    &= P+ \left( \Gamma\hat\Sigma^\dagger WJ_Y^{i-d-1}\mathbf a(\theta)\right)^2,
\end{align}
where the inequality follows from $\Xi_{i-d-1}\preceq W^\dagger\hat\Sigma(W^\dagger)^\top$ and the last equality from $\Gamma\hat\Sigma^\dagger\Gamma^\top=P$. Since $J_Y$ is Schur and $\mathbf a(\theta)$ is uniformly bounded, the second term is summable over refinement times $i$. The embedding energy is also uniformly bounded, so the total energy is at most $(n-d)P+K$ for some constant $K<\infty$, and a finite zero-padding can make it satisfy the power constraint.
\end{proof}

\section{Conclusions}\label{sec:conclusions}
For stationary colored Gaussian noise with rational PSDs, a
capacity-achieving stationary input can be generated by strictly causal linear filtering of the noise reconstructed from feedback. A finite message-dependent initialization realizes this structure as a
generalized SK coding scheme, with refinement parameters determined by the convex optimizer. For rational PSDs, this result supplies the
missing justification for Kim's spectral reduction
\cite{Kim2010,DerpichOstergaard2022}.

The scope of component removal beyond this setting remains to be characterized. For MIMO channels and nonstationary noise, the relevant question is under what conditions the state-space optimization admits an optimizer with $M=0$. For arbitrary stationary spectra, the corresponding question concerns the validity of $S_V\equiv0$ without a finite-dimensional state-space representation.

\bibliographystyle{IEEEtran}
\bibliography{refs}

@ARTICLE{5714269,
  author={Shayevitz, Ofer and Feder, Meir},
  journal={IEEE Transactions on Information Theory}, 
  title={Optimal Feedback Communication Via Posterior Matching}, 
  year={2011},
  volume={57},
  number={3},
  pages={1186-1222},
  doi={10.1109/TIT.2011.2104992}}

@ARTICLE{1054671,
  author={Kadota, T. and Zakai, M. and Ziv, J.},
  journal={IEEE Transactions on Information Theory}, 
  title={Capacity of a continuous memoryless channel with feedback}, 
  year={1971},
  volume={17},
  number={4},
  pages={372-378},
  doi={10.1109/TIT.1971.1054671}}

@article{Tiernan1974,
  author  = {James C. Tiernan and J. Pieter M. Schalkwijk},
  title   = {An Upper Bound to the Capacity of the Band-Limited {G}aussian Autoregressive Channel with Noiseless Feedback},
  journal = {IEEE Transactions on Information Theory},
  volume  = {20},
  number  = {3},
  pages   = {311--316},
  month   = may,
  year    = {1974},
  doi     = {10.1109/TIT.1974.1055231}
}

@article{Wolfowitz1975,
  author  = {Jacob Wolfowitz},
  title   = {Signalling over a {G}aussian Channel with Feedback and Autoregressive Noise},
  journal = {Journal of Applied Probability},
  volume  = {12},
  number  = {4},
  pages   = {713--723},
  month   = dec,
  year    = {1975},
  doi     = {10.2307/3212722}
}

@article{GallagerNakiboglu2010,
  author  = {Robert G. Gallager and Bar{\i}\c{s} Nakibo{\u{g}}lu},
  title   = {Variations on a Theme by {S}chalkwijk and {K}ailath},
  journal = {IEEE Transactions on Information Theory},
  volume  = {56},
  number  = {1},
  pages   = {6--17},
  year    = {2010},
  month   = jan,
  doi     = {10.1109/TIT.2009.2034896}
}

@book{KailathSayedHassibi2000,
  author    = {Kailath, Thomas and Sayed, Ali H. and Hassibi, Babak},
  title     = {Linear Estimation},
  publisher = {Prentice Hall},
  year      = {2000}
}

@article{CoverPombra1989,
  author={T. M. Cover and S. Pombra},
  title={Gaussian Feedback Capacity},
  journal={IEEE Transactions on Information Theory},
  volume={35},
  number={1},
  pages={37--43},
  year={1989},
  month={Jan.},
  doi={10.1109/18.42171}
}

@article{Kim2010,
  author={Y.-H. Kim},
  title={Feedback Capacity of Stationary {G}aussian Channels},
  journal={IEEE Transactions on Information Theory},
  volume={56},
  number={1},
  pages={57--85},
  year={2010},
  month={Jan.},
  doi={10.1109/TIT.2009.2037047}
}

@article{DerpichOstergaard2022,
  author={M. Derpich and J. {\O}stergaard},
  title={Comments on ``Feedback Capacity of Stationary {G}aussian Channels''},
  journal={IEEE Transactions on Information Theory},
volume = {70},
number = {3},
pages  = {1848--1851},
year   = {2024},
month  = mar,
  doi={10.1109/TIT.2022.3182270}
}

@INPROCEEDINGS{SabagKostinaHassibi2021,
  author={Sabag, Oron and Kostina, Victoria and Hassibi, Babak},
  booktitle={2021 IEEE International Symposium on Information Theory (ISIT)}, 
  title={Feedback Capacity of {MIMO} {G}aussian Channels}, 
  year={2021},
  volume={},
  number={},
  pages={7-12}
  }

@ARTICLE{SabagKostinaHassibi2023,
  author={Sabag, Oron and Kostina, Victoria and Hassibi, Babak},
  journal={IEEE Transactions on Information Theory}, 
  title={Feedback Capacity of {MIMO} {G}aussian Channels}, 
  year={2023},
  volume={69},
  number={10},
  pages={6121-6136},
  doi={10.1109/TIT.2023.3283570}}

@article{Gattami2019,
  author={A. Gattami},
  title={Feedback Capacity of {G}aussian Channels Revisited},
  journal={IEEE Transactions on Information Theory},
  volume={65},
  number={3},
  pages={1948--1960},
  year={2019},
  month={Mar.},
  doi={10.1109/TIT.2018.2879457}
}

@article{LiuHan2019,
  author={T. Liu and G. Han},
  title={Feedback Capacity of Stationary {G}aussian Channels Further Examined},
  journal={IEEE Transactions on Information Theory},
  volume={65},
  number={4},
  pages={2492--2506},
  year={2019},
  month={Apr.},
  doi={10.1109/TIT.2018.2887241}
}

@ARTICLE{LiElia,
  author={Li, Chong and Elia, Nicola},
  journal={IEEE Transactions on Information Theory}, 
  title={Youla Coding and Computation of {G}aussian Feedback Capacity}, 
  year={2018},
  volume={64},
  number={4},
  pages={3197-3215}
  }

@book{Kailath2000,
  author={T. Kailath and A. H. Sayed and B. Hassibi},
  title={Linear Estimation},
  publisher={Prentice Hall},
  address={Upper Saddle River, NJ},
  year={2000}
}

@article{Shannon1956,
  author={C. E. Shannon},
  title={The Zero Error Capacity of a Noisy Channel},
  journal={IRE Transactions on Information Theory},
  volume={2},
  number={3},
  pages={8--19},
  year={1956},
  month={Sep.},
  doi={10.1109/TIT.1956.1056798}
}

@article{SchalkwijkKailath1966,
  author={J. P. M. Schalkwijk and T. Kailath},
  title={A Coding Scheme for Additive Noise Channels with Feedback---Part I: No Bandwidth Constraint},
  journal={IEEE Transactions on Information Theory},
  volume={12},
  number={2},
  pages={172--182},
  year={1966},
  month={Apr.}
}

@INPROCEEDINGS{RawatElia2021,
  author={Rawat, Abhishek and Elia, Nicola},
  booktitle={2020 IEEE Information Theory Workshop (ITW)}, 
  title={Feedback Capacity of {ISI} {MIMO} Channel with Colored Noise}, 
  year={2021},
  volume={},
  number={},
  pages={1-5},
  doi={10.1109/ITW46852.2021.9457644}}

@incollection{Chonavel2002,
  author    = {Chonavel, Thierry},
  title     = {Rational Spectral Densities},
  booktitle = {Statistical Signal Processing},
  series    = {Advanced Textbooks in Control and Signal Processing},
  publisher = {Springer},
  address   = {London},
  year      = {2002},
  pages     = {111--117}
}

@INPROCEEDINGS{SabagISIT,
  author={Sabag, Oron and Kostina, Victoria and Hassibi, Babak},
  booktitle={2022 IEEE International Symposium on Information Theory (ISIT)}, 
  title={Feedback Capacity of {G}aussian Channels with Memory}, 
  year={2022},
  volume={},
  number={},
  pages={2547-2552},
  doi={10.1109/ISIT50566.2022.9834799}}

@ARTICLE{Butman1969,
  author={Butman, S.},
  journal={IEEE Transactions on Information Theory}, 
  title={A general formulation of linear feedback communication systems with solutions}, 
  year={1969},
  volume={15},
  number={3},
  pages={392-400},
  doi={10.1109/TIT.1969.1054302}}

@ARTICLE{Butman1976,
  author={Butman, S.},
  journal={IEEE Transactions on Information Theory}, 
  title={Linear feedback rate bounds for regressive channels (Corresp.)}, 
  year={1976},
  volume={22},
  number={3},
  pages={363-366},
  doi={10.1109/TIT.1976.1055548}}

@ARTICLE{Ebert1970,
  author={Ebert, P. M.},
  journal={The Bell System Technical Journal}, 
  title={The capacity of the {G}aussian channel with feedback}, 
  year={1970},
  volume={49},
  number={8},
  pages={1705-1712},
  doi={10.1002/j.1538-7305.1970.tb04286.x}}

@INPROCEEDINGS{Shahar-Doron2004,
  author={Shahar-Doron, A. and Feder, M.},
  booktitle={International Symposium on Information Theory, 2004. ISIT 2004. Proceedings.}, 
  title={On a capacity achieving scheme for the colored {G}aussian channel with feedback}, 
  year={2004},
  volume={},
  number={},
  pages={74},
  doi={10.1109/ISIT.2004.1365109}}

@ARTICLE{Elia2004,
  author={Elia, N.},
  journal={IEEE Transactions on Automatic Control}, 
  title={When {B}ode meets {S}hannon: control-oriented feedback communication schemes}, 
  year={2004},
  volume={49},
  number={9},
  pages={1477-1488},
  doi={10.1109/TAC.2004.834119}}

@ARTICLE{Ozarow1990,
  author={Ozarow, L.H.},
  journal={IEEE Transactions on Information Theory}, 
  title={Random coding for additive Gaussian channels with feedback}, 
  year={1990},
  volume={36},
  number={1},
  pages={17-22},
  doi={10.1109/18.50369}}

@ARTICLE{Yang2007,
  author={Yang, Shaohua and Kav\v{c}i\'{c}, Aleksandar and Tatikonda, Sekhar},
  journal={IEEE Transactions on Information Theory}, 
  title={On the Feedback Capacity of Power-Constrained Gaussian Noise Channels With Memory}, 
  year={2007},
  volume={53},
  number={3},
  pages={929-954},
  doi={10.1109/TIT.2006.890728}}

@ARTICLE{Kim2006,
  author={Young-Han Kim},
  journal={IEEE Transactions on Information Theory}, 
  title={Feedback capacity of the first-order moving average {G}aussian channel}, 
  year={2006},
  volume={52},
  number={7},
  pages={3063-3079},
  doi={10.1109/TIT.2006.876217}}

@ARTICLE{Ihara2012,
  author={Ihara, Shunsuke},
  journal={IEEE Transactions on Information Theory}, 
  title={Upper Bounds of Error Probabilities for Stationary {G}aussian Channels With Feedback}, 
  year={2012},
  volume={58},
  number={12},
  pages={7068-7072},
  doi={10.1109/TIT.2012.2211335}}

\end{document}